\documentclass[mnsc,nonblindrev]{informs3-ssrn}
\RequirePackage{tgtermes}
\RequirePackage{newtxtext}
\RequirePackage{newtxmath}
\RequirePackage{bm}
\RequirePackage{endnotes}

\OneAndAHalfSpacedXI 

\usepackage[utf8]{inputenc}

\usepackage{algorithm}
\usepackage{algpseudocode}
\usepackage{amsmath}
\usepackage{amssymb}
\usepackage{subcaption}
\usepackage{array}
\usepackage{todonotes}
\usepackage{longtable}
\usepackage{tikz}
\usepackage{pgfplots}
\usepackage{titlesec}

\titlespacing*{\section}
{0pt}          
{1.5ex plus 1ex minus .2ex} 
{1.5ex plus .2ex}           

\usepackage{xcolor}

\usepackage{natbib}
 \bibpunct[, ]{(}{)}{,}{a}{}{,}%
 \def\bibfont{\small}%
\let\cite\citep

\EquationsNumberedThrough    

\TheoremsNumberedThrough     
\ECRepeatTheorems  %

\MANUSCRIPTNO{}

\usepackage{algorithm}
\usepackage{algpseudocode}
\usepackage{multirow}
\usepackage{multicol}
\usepackage{mathtools}

\usepackage[utf8]{inputenc}
\usepackage{pgfplots}
\usepgfplotslibrary{groupplots,dateplot}
\usetikzlibrary{patterns,shapes.arrows}
\pgfplotsset{compat=newest}
\pgfplotsset{scaled y ticks=false}
\usepackage{caption}
\usepackage{subcaption}
\usepackage{xcolor}
\usepackage{soul}
\usepackage{todonotes}
\usepgfplotslibrary{groupplots,dateplot}
\usetikzlibrary{patterns,shapes.arrows}
\pgfplotsset{compat=newest}

\usepackage{comment}
\usepackage{booktabs} 

\usepackage{url}
\usepackage[capitalize,noabbrev]{cleveref}

\usepackage{cleveref}
\crefname{problem}{Problem}{Problems}
\Crefname{problem}{Problem}{Problems}


\usepackage{float}

\begin{document}



\RUNTITLE{Budget Generalized Fractional Influence Maximization}

\TITLE{Optimizing Marketing Budgets for Social Influence: A Generalized Fractional Incentives Framework}

\ARTICLEAUTHORS{%
\AUTHOR{Abhishek K. Umrawal\footnote{Abhishek K. Umrawal is the corresponding author.}}
\AFF{
University of Illinois Urbana-Champaign, Urbana, IL 61801, USA, \EMAIL{aumrawal@illinois.edu}}
\AUTHOR{Akhil Bhimaraju}
\AFF{
University of Illinois Urbana-Champaign, Urbana, IL 61801, USA, \EMAIL{akhilb3@illinois.edu}} 
\AUTHOR{Eliot W. Robson
}
\AFF{Narmi, New York, NY 10016, USA, \EMAIL{eliot.robson24@gmail.com}}
\AUTHOR{Lav R. Varshney
}
\AFF{
Stony Brook University, Stony Brook, NY 11794 USA, \EMAIL{lav.varshney@stonybrook.edu}}
} 

\ABSTRACT{%
Marketing managers routinely allocate limited budgets across individuals in social networks whose responsiveness and influence vary substantially. We study the problem of Budgeted Generalized Fractional Influence Maximization (BGFIM), which captures the allocation of continuous levels of marketing spend across users in the presence of heterogeneous response and network spillovers. In this framework, adoption probabilities are governed by user-specific response functions that capture both baseline propensity and marginal responsiveness, while incentive costs vary across users. We propose a scalable solution approach based on a discrete greedy procedure with provable approximation guarantees, enabling tractable optimization under realistic budget constraints. Our formulation generalizes prior influence maximization models by allowing heterogeneous response and cost structures, providing a more flexible representation of targeted marketing interventions. Using experiments on multiple large-scale networks, we show that continuous allocation strategies can outperform standard heuristics and improve influence outcomes under heterogeneous conditions, and we illustrate how the framework can support profit-oriented decision making by identifying budget levels that balance marginal returns and costs. Overall, our results highlight the value of fine-grained budget allocation in networked marketing environments and offer a scalable approach for data-driven marketing resource allocation.
}%



\maketitle


\section{Introduction}
\label{sec:introduction}

The rapid proliferation of online social networks has transformed the landscape of modern communication and commerce, positioning digital platforms as central hubs for consumer engagement and product promotion \citep{aral2013social, holt2016branding, lobel2017customer, rehman2022imc}. This evolution has catalyzed the emergence of a multi-billion-dollar ``influencer economy,'' wherein firms strategically allocate substantial marketing budgets to incentivize individuals—often social media influencers—to adopt and advocate for products, services, or ideas \citep{gupta2008allocating, aral2020social, pan2025influencer}. Influencer marketing has proven to be a powerful tool for shaping consumer attitudes and purchase intentions, with effectiveness driven by factors such as influencer credibility, social identity, and content value \citep{vrontis2021social, pan2025influencer}. A growing body of research in marketing and operations has examined the role of networks in shaping adoption and optimal targeting decisions, highlighting the importance of accounting for both individual responsiveness and network spillovers in designing effective campaigns \citep{goldenberg2001using, libai2013customer, mallipeddi2022framework, kanuri2018scheduling}. These patterns are also reflected in practitioner reports and industry examples, which illustrate how influencer campaigns often combine creator--brand fit, influencer seeding, and targeted incentives to drive engagement and product adoption \citep{getsaral2023graza,naceva2026campaigns,modash2025campaigns}. In this complex and rapidly evolving environment, the efficient allocation of finite marketing resources across individuals in a network has become a pressing managerial challenge, demanding data-driven approaches to maximize campaign reach and return on investment \citep{gupta2008allocating, lang2022targeted}.

To promote their products and services in such networked environments, firms must decide which individuals to target and how to allocate incentives so as to maximize downstream diffusion \citep{kempe2003maximizing, aral2012influential}. This decision problem has been studied in the literature on influence maximization, which models how the selection of initial adopters can drive cascading adoption through social networks \citep{kempe2003maximizing, goyal2011data, borgs2014maximizing}. In many practical settings, firms provide monetary incentives or discounts to encourage adoption and advocacy, with the likelihood of engagement increasing with the level of the incentive. However, marketing budgets are inherently limited, requiring firms to balance the potential influence of individuals against the costs of targeting them. This trade-off between influence potential and cost has motivated a substantial body of work on budget-constrained influence maximization \citep{goyal2011data, borgs2014maximizing}. A common assumption in this literature is that incentives are binary—individuals are either fully targeted or not at all—which fails to capture the more nuanced strategies observed in modern marketing campaigns, where firms frequently employ partial or personalized incentives \citep{Vondrak2008, calinescu2011maximizing, yang2016continuous, chen2020scalable}.

While allowing fractional or continuous incentives provides a more flexible alternative to binary targeting, existing approaches continue to rely on restrictive modeling assumptions that limit their applicability in marketing settings. In particular, prior work often assumes a direct equivalence between the level of an incentive and the probability of user adoption \citep{umrawal2023fractional,umrawal2026jair}, effectively treating incentives and behavioral responses as interchangeable. This assumption is difficult to justify in practice, as it ignores substantial heterogeneity in consumer behavior: some individuals may adopt with minimal incentives due to inherent affinity, while others may remain unresponsive even under substantial incentives. For example, a dedicated early adopter might activate with a 5\% discount, while a brand-loyal customer of a competing product may remain unresponsive even under much larger incentives (e.g., 20\%). These behavioral differences are well-documented in the marketing literature, which emphasizes the importance of consumer segmentation and differential responsiveness to promotional stimuli \citep{goldenberg2001talk, libai2013customer, aral2013social}. Furthermore, these models do not adequately account for variation in the cost of targeting different individuals, which is a central feature of real-world marketing campaigns. As emphasized in the marketing literature, differences in baseline preferences, responsiveness to incentives, and economic costs play a critical role in determining the effectiveness of targeted interventions \citep{goldenberg2001using, libai2013customer}. These limitations highlight the need for models that can jointly capture heterogeneous responses to incentives and heterogeneous costs within a unified framework for budget allocation in networks.

To address these limitations, we study the problem of Budgeted Generalized Fractional Influence Maximization (BGFIM), which models the allocation of marketing resources across individuals in a network in the presence of heterogeneous response and network spillovers. In our framework, firms assign continuous levels of incentives to users subject to a budget constraint, and adoption probabilities are governed by user-specific response functions that capture both baseline propensity and marginal responsiveness. This formulation separates the economic decision of how much incentive to allocate from the behavioral response it induces, allowing for a more realistic representation of consumer heterogeneity. In addition, the model incorporates heterogeneous costs of targeting individuals, reflecting differences in influencer compensation and the resources required to engage different individuals.
We propose a scalable solution approach based on a discrete greedy procedure with provable approximation guarantees, enabling tractable optimization in large-scale networks. We evaluate the performance of the approach through experiments on multiple real-world networks and further illustrate how the framework can be used to support profit-oriented decision making by identifying budget levels that balance marginal returns and costs.

\subsection{Contributions}

This paper makes three main contributions.

\begin{enumerate}
    \item \textbf{Decision framework for budget allocation in networks.} 
    We formulate the Budgeted Generalized Fractional Influence Maximization (BGFIM) problem, which captures the challenge of allocating limited marketing resources across individuals in a network in the presence of heterogeneous response and network spillovers. By allowing continuous levels of incentives, heterogeneous response functions, and heterogeneous costs, the framework provides a more realistic representation of targeted marketing interventions than existing binary or homogeneous models.

    \item \textbf{Scalable solution approach with theoretical guarantees.} 
    We propose a discrete greedy algorithm that provides a tractable approach to solving the BGFIM problem at scale and establishes provable approximation guarantees for its performance. The method is designed to account jointly for influence generation and incentive costs, enabling efficient optimization under realistic budget constraints.

    \item \textbf{Quantitative insights for marketing resource allocation.} 
    Through experiments on multiple large-scale networks, we demonstrate that continuous allocation strategies can outperform standard heuristics, particularly in settings with heterogeneous response and costs. We further show how the framework can be used to support profit-oriented decision making by identifying budget levels that balance marginal returns and costs, providing actionable guidance for marketing campaign design.
\end{enumerate}

\subsection{Organization}
The remainder of the paper is organized as follows. Section~\ref{sec:related-work} reviews the conceptual motivation and related literature. Section~\ref{sec:model} introduces the modeling framework, notation, and the Budgeted Generalized Fractional Influence Maximization (BGFIM) problem. Section~\ref{sec:algorithm} presents our solution approach, including the discrete greedy algorithm and its theoretical approximation guarantees. Section~\ref{sec:experiments} reports computational experiments on real-world social networks and examines the use of nested solutions for profit-oriented budget selection. Finally, Section~\ref{sec:conclusion} concludes and outlines directions for future research.

\vspace{.1in}
\section{Conceptual Motivation and Related Literature}
\label{sec:related-work}

Diffusion and word-of-mouth processes are foundational to marketing science, where they are used to study how products, behaviors, and information spread through populations. Classic models, such as the Bass diffusion framework \citep{bass1969new} and its network-based extensions \citep{goldenberg2001using}, highlight how individual adoption decisions aggregate into population-level dynamics through both external influence and interpersonal interactions. Subsequent work has emphasized the central role of social networks in shaping these processes, demonstrating that peer effects, social reinforcement, and network structure critically determine the speed and extent of diffusion \citep{libai2013customer, aral2013social, lobel2017customer}. In modern digital environments, these mechanisms are further amplified by online platforms, where repeated exposure, homophily, and algorithmic content curation influence both adoption behavior and downstream word-of-mouth dynamics \citep{aral2020social, pan2025influencer}. Collectively, this body of work underscores that diffusion is inherently a networked phenomenon, in which both individual heterogeneity and interconnected social influence must be accounted for to accurately describe and predict market outcomes.

Building on this foundation, a central concern in marketing is how firms can actively shape diffusion outcomes through targeted interventions. A growing literature examines how firms allocate resources across consumers and influencers to maximize engagement, adoption, and long-term value \citep{libai2013customer, gupta2008allocating, mallipeddi2022framework}. These studies emphasize that effective marketing strategies must account not only for who is targeted, but also for when, how, and with what intensity interventions are deployed. Empirical and field-based research highlights the importance of timing, exposure frequency, and content design in shaping consumer response, particularly in digital environments where users are repeatedly exposed to competing signals \citep{kanuri2018scheduling, pan2025influencer}. Moreover, the effectiveness of targeted interventions depends critically on network spillovers: influencing one individual can generate indirect effects through their social connections, amplifying or attenuating the overall impact of marketing actions. These insights collectively point to the need for decision frameworks that integrate individual-level targeting with network-mediated spillovers, allowing firms to allocate limited marketing budgets in a manner that accounts for both direct and indirect effects.

To operationalize marketing interventions in networked settings, researchers have long studied how behaviors and information spread through social contagion processes. Early work in sociology and marketing models adoption as arising from interpersonal influence and word-of-mouth interactions, where individuals are affected by the behavior of their peers \citep{granovetter1978threshold, schelling2006micromotives, goldenberg2001talk, goldenberg2001using}. Building on these insights, \citet{domingos2001mining} introduced the idea of leveraging network structure for targeted marketing, framing customer selection as a problem of maximizing network value. This perspective was formalized by \citet{kempe2003maximizing}, who introduced the influence maximization (IM) problem and provided a unified framework for studying how targeted interventions can maximize downstream diffusion in networks.

At the core of this framework are canonical diffusion models that formalize how influence propagates through networks. The Independent Cascade (IC) and Linear Threshold (LT) models, introduced in \citet{kempe2003maximizing}, provide two widely used representations of diffusion dynamics. The IC model captures stochastic contagion through pairwise interactions, closely related to word-of-mouth diffusion processes studied in marketing \citep{goldenberg2001talk, goldenberg2001using}, whereas the LT model captures cumulative peer influence through threshold-based adoption rules rooted in sociological models of collective behavior \citep{granovetter1978threshold, schelling2006micromotives}. Under these models, the expected influence function is monotone and submodular, enabling greedy algorithms with strong approximation guarantees \citep{kempe2003maximizing, nemhauser1978analysis}. More recent extensions enrich these models to better capture behavioral dynamics, including reinforcement effects such as the Pressure Threshold (PT) model \citep{stutsman2026pressure,umrawal2026social}.

Building on these diffusion models, a substantial body of work in computer science and operations research has developed scalable algorithms for influence maximization. Early contributions focused on improving the efficiency of greedy methods, including the CELF algorithm and its variants \citep{celf-07, goyal2011celf++}, followed by sampling-based techniques that enable near-linear-time approximation on large networks \citep{borgs2014maximizing, tang2015influence}. Subsequent work has incorporated structural properties of networks to further improve scalability and solution quality, including community-aware methods and decomposition-based approaches \citep{umrawal2023leveraging, umrawal2023community, robson2025community}. In parallel, extensions to adaptive and online settings have been developed to account for uncertainty, learning, and dynamic environments \citep{agarwal2022stochastic, nie2022explore, xu2026lofa}. Collectively, this literature provides powerful optimization tools for implementing network-based marketing strategies at scale.

Complementing these model-based approaches, a parallel stream of research has developed data-driven methods that infer influence and model user response directly from observed behavior. In network science, credit distribution approaches estimate influence by attributing observed downstream adoption to upstream nodes based on cascade data, avoiding reliance on explicit generative diffusion assumptions \citep{goyal2011data}. In marketing and operations, related work models engagement as a function of exposure rather than propagation. For example, \citet{mallipeddi2022framework} develops an empirical framework for influencer marketing in which user response depends on cumulative exposure to multiple influencers, incorporating effects such as repetition and temporal decay. These data-driven approaches provide greater empirical realism by capturing exposure effects and behavioral heterogeneity, but they typically abstract away from explicit modeling of incentive allocation and economic decision-making in networked settings.

Beyond binary targeting decisions, recent research has sought to better align influence maximization models with the realities of marketing practice by allowing for more flexible allocation of resources. In particular, a growing body of work considers fractional or continuous variants of the problem, in which firms can assign varying levels of incentives or effort across individuals rather than making purely discrete selections \citep{yang2016continuous, chen2020scalable, BhimarajuRVU2024}. These formulations capture the intuition that marketing interventions—such as discounts, payments to influencers, or promotional efforts—can be adjusted continuously to influence consumer behavior. Algorithmically, these models introduce additional complexity, requiring optimization over high-dimensional continuous decision spaces while preserving tractable approximation guarantees. Recent advances have addressed these challenges through discretization techniques, lattice-based methods, and scalable greedy approximations \citep{chen2020scalable, umrawal2023fractional}. Related work has also explored extensions that incorporate structural features such as community organization and heterogeneous influence pathways, further improving the effectiveness of network interventions in practice \citep{umrawal2023leveraging, umrawal2023community, robson2025community}.

Despite these advances, existing fractional formulations rely on simplifying assumptions that limit their applicability in marketing contexts. A common modeling approach assumes a direct mapping between the level of incentive assigned to an individual and their probability of adoption, effectively treating economic inputs and behavioral responses as interchangeable \citep{umrawal2023fractional,umrawal2026jair}. This abstraction overlooks well-documented heterogeneity in consumer behavior, including variation in baseline preferences, responsiveness to incentives, and sensitivity to social influence \citep{goldenberg2001using, libai2013customer}. Moreover, most existing models do not explicitly account for heterogeneous costs of targeting individuals, despite their central importance in real-world settings such as influencer compensation, customer acquisition, and platform-specific pricing. These limitations highlight a gap between existing optimization frameworks and the behavioral and economic realities of modern marketing environments.

Taken together, these streams of research highlight both the strengths and the limitations of existing approaches to network-based marketing. Prior work provides a rich set of tools for modeling diffusion, designing targeted interventions, and optimizing influence in large-scale networks. However, a key gap remains in integrating the economic and behavioral dimensions of marketing decision-making within these frameworks. In particular, existing models do not fully capture the joint role of heterogeneous response functions, continuous incentive allocation, and heterogeneous costs in shaping diffusion outcomes. 

Motivated by these limitations, we develop a framework that more closely aligns influence optimization with the realities of modern marketing environments. Our approach builds on the influence maximization paradigm while generalizing it to allow for flexible incentive-response relationships and heterogeneous cost structures across individuals. By explicitly separating the allocation decision from the resulting behavioral response, the framework captures the interplay between economic incentives and consumer heterogeneity in a networked setting. This perspective bridges insights from marketing, network science, and optimization, providing a foundation for studying how firms can allocate limited resources to maximize both direct and spillover effects in social networks.
\vspace{.1in}
\section{Mathematical Model and Problem Formulation}
\label{sec:model}

Marketing interventions in social networks rely on interpersonal influence, where adoption by one individual can affect the behavior of others through interactions such as recommendations, content sharing, and peer endorsement, as reflected in observable engagement actions (e.g., likes, shares, or reposts), giving rise to diffusion across the network. To capture these influence pathways, we model the social network as a directed graph $G = (V, E)$, where each node $v \in V$ represents an individual user and each directed edge $(v, u) \in E$ indicates a potential influence relationship from user $v$ to user $u$. That is, if $(v, u) \in E$, then user $v$ may affect the likelihood that user $u$ adopts a product or behavior. We associate each edge $(v,u) \in E$ with an influence probability $e_{v,u} \in [0,1]$, capturing the strength of the relationship between users. These probabilities may reflect factors such as interaction frequency, trust, or historical engagement, and in practical settings can be estimated from observed diffusion or interaction data \citep{goyal2010learning, gomezrodriguez2010inferring}. This graph-based abstraction reflects the structure of online platforms, where connections encode pathways through which influence can spread; the neighborhood of a node $v$, denoted by $\mathcal{N}_v = \{u \in V : (v, u) \in E\}$, consists of all users that $v$ can directly influence.

As discussed in Section~\ref{sec:related-work}, diffusion models provide a framework for describing how adoption and information propagate through social networks. To model how adoption propagates in our setting, we adopt the independent cascade (IC) diffusion model \citep{kempe2003maximizing}, which provides a tractable and widely used representation of stochastic influence spread. Under this model, time evolves in discrete steps $u = 0,1,2,\dots$, and once a user becomes active, they are given a single opportunity to influence each inactive neighbor independently with probability $e_{v,u}$. If the activation attempt succeeds, the neighbor becomes active in the next time step; otherwise, no further attempts are made along that edge. Starting from an initial set of active users $S \subseteq V$ at time $u=0$, the diffusion process unfolds through successive rounds of such influence attempts until no further activations occur.

To formalize the diffusion process, let
\[
Y_u^{(v)} =
\begin{cases}
1, & \text{if node $v$ is active at time $u$,} \\
0, & \text{otherwise.}
\end{cases}
\]
The diffusion is \textit{progressive}, meaning that nodes can switch from inactive to active, but not back.

\begin{definition}[Influence Function]
Given a seed set $S \subseteq V$, the influence $\sigma(S)$ is defined as the expected number of active nodes at the end of the diffusion process:
\begin{align*}
\sigma(S) = \mathbb{E}\left[\sum_{v \in V} Y_T^{(v)} \;\middle|\; Y_0^{(v)} = 1 \ \forall v \in S,\ Y_0^{(v)} = 0 \ \forall v \notin S \right],
\end{align*}
where $T$ denotes the (random) time at which the diffusion process terminates.
\end{definition}

The influence function $\sigma: 2^V \to \mathbb{R}_{\ge 0}$ captures the total reach generated by targeting an initial set of users.

\begin{definition}[Monotonicity]
A set function $f: 2^V \to \mathbb{R}_{\ge 0}$ is monotone if for any $A \subseteq B \subseteq V$, we have $f(A) \le f(B)$.
\end{definition}

\begin{definition}[Submodularity]
A set function $f: 2^V \to \mathbb{R}_{\ge 0}$ is submodular if for any $A \subseteq B \subseteq V$ and any node $v \notin B$, it holds that
\begin{align*}
f(A \cup \{v\}) - f(A) \ge f(B \cup \{v\}) - f(B).
\end{align*}
\end{definition}

\begin{lemma}[\citet{kempe2003maximizing}]
Under the IC model, the influence function $\sigma(\cdot)$ is monotone and submodular.
\end{lemma}

Monotonicity reflects the idea that targeting additional users cannot reduce the expected spread of influence. Submodularity captures diminishing returns: the marginal benefit of adding a new user is larger when few users have already been targeted. In marketing terms, early targeted users can reach previously untapped segments of the network, whereas additional users increasingly generate overlapping influence.

\begin{problem}[Discrete Influence Maximization (DIM)] \label{problem:DIM}
The classical influence maximization problem seeks to identify a seed set $S \subseteq V$ of size at most $k$ that maximizes the influence function $\sigma(S)$:
\begin{align*}
\argmax_{S \subseteq V} \quad & \sigma(S) \\
\text{s.t.} \quad & |S| \le k.
\end{align*}
\end{problem}
\noindent
Problem~\ref{problem:DIM} is known to be NP-hard \citep{kempe2003maximizing}, reflecting the combinatorial complexity of selecting an optimal set of users in large networks.

While the discrete formulation provides important structural insights, it abstracts away from key features of practical marketing interventions. In many applications, firms do not make purely binary targeting decisions; instead, they allocate varying levels of incentives across users, resulting in heterogeneous and probabilistic adoption behavior \citep{domingos2001mining, goldenberg2001talk, goldenberg2001using, aral2009distinguishing, iyengar2011opinion, yoganarasimhan2012impact, yang2016continuous}. Such practices are commonly observed in referral programs and targeted promotional campaigns, where the intensity of incentives influences the likelihood of user participation \citep{ryu2007referral, schmitt2011referrals}. This motivates the development of more general fractional models that accommodate continuous decision variables and user-specific response functions.

Motivated by the need to model heterogeneous and non-binary marketing interventions, we generalize the discrete formulation by allowing continuous decision variables. Specifically, instead of selecting a binary seed set $S \subseteq V$, we consider a vector $\mathbf{x} = (x_1, \ldots, x_n) \in [0,1]^n$, where $x_v$ denotes the level of incentive or marketing effort allocated to user $v$. 

We assume that each user $v \in V$ becomes active at time $u=0$ with probability $f_v(x_v)$, where $f_v: [0,1] \to [0,1]$ is a non-decreasing function capturing how incentives translate into initial adoption probabilities. This formulation allows for heterogeneous user responses to marketing interventions, reflecting differences in preferences, responsiveness, and baseline propensity to adopt. The subsequent diffusion process then evolves according to the same independent cascade (IC) dynamics as in the discrete setting.

\begin{definition}[Continuous Influence Function]
Given an allocation vector $\mathbf{x} \in [0,1]^n$, the expected influence is defined as
\begin{align*}
\tilde{\sigma}(\mathbf{x}) = \mathbb{E}\left[\sum_{v \in V} Y_T^{(v)} \;\middle|\; Y_0^{(v)} = 1 \text{ with probability } f_v(x_v) \text{ independently for each } v \in V \right].
\end{align*}
\end{definition}

The function $\tilde{\sigma}: [0,1]^n \to \mathbb{R}_{\ge 0}$ captures the expected spread of influence resulting from heterogeneous and probabilistic seeding induced by continuous marketing actions.


An equivalent and useful interpretation expresses the model in terms of a distribution over random seed sets. Let $S_0 \subseteq V$ denote the random set of nodes activated at time $0$. Then the expected influence can be written as
\[
\tilde{\sigma}(\mathbf{x}) = \mathbb{E}_{S_0}\left[\sigma(S_0) \mid \mathbf{x}\right],
\]
where $\sigma(S_0)$ denotes the expected cascade size under the IC model starting from seed set $S_0$. This perspective highlights that the continuous formulation can be viewed as inducing a probabilistic relaxation of the discrete influence maximization problem.

The general formulation above allows for flexibility in modeling how incentives translate into adoption probabilities. A natural and widely studied special case assumes a linear relationship, i.e., $f_v(x_v) = x_v$ for all $v \in V$. Under this assumption, each component $x_v$ can be interpreted directly as the probability that user $v$ is initially activated, and the random seed set $S_0$ is generated by including each $v$ independently with probability $x_v$. In this case, the expected influence function $\tilde{\sigma}(\mathbf{x})$ coincides with the multilinear extension of the discrete influence function $\sigma(\cdot)$ associated with the independent cascade model \citep{umrawal2023fractional,umrawal2026jair}. 

While this identity-based formulation admits strong algorithmic guarantees, it imposes a restrictive assumption that the level of incentive directly corresponds to the probability of adoption and that this relationship is identical across users.

More generally, we consider affine response functions of the form $f_v(x_v) = a_v x_v + b_v$, where $a_v \ge 0$ and $b_v \in [0,1]$ are user-specific parameters. In this work, we focus on this class of activation functions and develop a polynomial-time approximation algorithm with a constant-factor guarantee for this generalized setting.

This formulation captures two key dimensions of user heterogeneity:
\begin{enumerate}
    \item The parameter $b_v$ represents the \emph{base activation probability} of user $v$, i.e., the likelihood that the user adopts the product even in the absence of any incentive. This accounts for organic interest or pre-existing affinity toward the product.
    
    \item The parameter $a_v$ quantifies the \emph{susceptibility} of user $v$ to incentives. A higher value of $a_v$ indicates that the user's activation probability increases rapidly with additional marketing effort, reflecting high price sensitivity or promotional responsiveness. Conversely, a lower $a_v$ suggests that the user is less influenced by incentives, perhaps due to brand loyalty or lack of interest.
\end{enumerate}

By allowing both baseline adoption and marginal responsiveness to vary across users, the affine formulation extends the identity model to capture richer and more realistic patterns of user behavior in networked diffusion processes.

From a modeling perspective, the affine structure provides a simple yet expressive parametric family that balances flexibility with analytical tractability. In particular, it enables the incorporation of heterogeneous user behavior while avoiding more complex nonlinear response functions that may hinder optimization. From an applied standpoint, the parameters $a_v$ and $b_v$ can be estimated from historical user engagement and campaign data, for example, by analyzing how observed adoption probabilities vary with different levels of incentives across users.

Prior work in marketing and social networks has emphasized the importance of heterogeneous influence and adoption behavior \citep{aral2009distinguishing, iyengar2011opinion}, while continuous and discretized influence maximization frameworks \citep{yang2016continuous, chen2020scalable} allow for heterogeneous response functions across users. Our formulation makes this heterogeneity explicit and interpretable, providing a natural bridge between data-driven modeling and tractable network optimization.

Furthermore, in many practical settings, marketing interventions are subject to budget constraints and heterogeneous costs across users. The cost of allocating incentives to a user may depend on factors such as their reach, engagement level, or alignment with the product, making it unrealistic to assume uniform cost structures. To capture this, we associate each user $v \in V$ with a cost parameter $w_v > 0$, which represents the cost per unit of incentive allocated to that user.





Given a total budget $k \in \mathbb{R}_{\ge 0}$, the allocation vector $\mathbf{x}$ must satisfy a budget constraint that accounts for heterogeneous user costs. For vectors $\mathbf{x}, \mathbf{w} \in \mathbb{R}^n$, we denote the inner product by
\[
\langle \mathbf{w}, \mathbf{x} \rangle = \sum_{v \in V} w_v x_v,
\]
so that the budget constraint can be written compactly as
\[
\langle \mathbf{w}, \mathbf{x} \rangle \le k.
\]

Let $k \in \mathbb{R}_{\ge 0}$ denote the total available budget. The objective is to select an allocation vector $\mathbf{x} = (x_1,\dots,x_n) \in [0,1]^n$ that maximizes the expected spread of influence, subject to a budget constraint. 

The decision maker seeks to allocate incentives across users to maximize the expected spread of influence subject to the budget constraint. This leads to the following optimization problem, which we refer to as the Budgeted Generalized Fractional Influence Maximization (BGFIM) problem:

\begin{problem}[Budgeted Generalized Fractional Influence Maximization (BGFIM)]
\label{problem:bgfim}
\begin{align*}
    \argmax_{\mathbf{x} \in [0,1]^n} \quad & \tilde{\sigma}(\mathbf{x}),\\
    \text{s.t.} \quad &\langle \mathbf{w}, \mathbf{x} \rangle \le k.
\end{align*}
\end{problem}

This formulation generalizes several models studied in the literature. In the case where the cost parameters are uniform, i.e., $w_v = 1$ for all $v \in V$, the problem reduces to a generalized fractional influence maximization (GFIM) framework with heterogeneous activation functions \cite{BhimarajuRVU2024}. Furthermore, when the activation functions take the identity form $f_v(x_v) = x_v$, and the cost parameters are uniform, the formulation recovers the fractional influence maximization (FIM) problem \citep{umrawal2023fractional,umrawal2026jair}.

By allowing both heterogeneous cost parameters and affine activation functions, the proposed BGFIM formulation strictly generalizes these prior models and provides a more expressive framework for modeling incentive-driven diffusion in social networks.

\vspace{.1in}
\section{Efficient Approximation Algorithm with Provable Guarantees}
\label{sec:algorithm}

We now turn to the problem of solving the Budgeted Generalized Fractional Influence Maximization (BGFIM) problem defined in \Cref{problem:bgfim}. In the classical discrete influence maximization setting, the objective function is monotone and submodular, enabling the use of efficient greedy algorithms with strong approximation guarantees \citep{kempe2003maximizing, nemhauser1978analysis}. More generally, submodular maximization under cardinality and knapsack constraints has been extensively studied, with both discrete and continuous variants admitting $(1 - 1/e)$-approximation algorithms \citep{Vondrak2008, calinescu2011maximizing}. 

However, extending these techniques to the BGFIM setting presents significant challenges. In particular, the objective function $\tilde{\sigma}(\mathbf{x})$ is defined over a continuous domain and depends on user-specific activation functions and heterogeneous cost parameters. Unlike the identity case, where the objective reduces to the multilinear extension of a submodular function, the presence of affine activation functions alters this structure and complicates the direct application of existing optimization methods. As a result, designing scalable algorithms with provable guarantees for \Cref{problem:bgfim} requires new techniques that can effectively handle both continuous decision variables and heterogeneous user behavior.

Our approach builds on the observation that, despite the loss of a standard multilinear structure, the objective function $\tilde{\sigma}(\mathbf{x})$ still retains a form of smooth marginal behavior that can be exploited algorithmically. In particular, the impact of allocating a small amount of incentive to a user can be characterized through the local sensitivity of their activation function. For affine functions of the form $f_v(x_v) = a_v x_v + b_v$, this sensitivity is captured by the parameter $a_v$, which determines how rapidly the user’s activation probability increases with additional incentives.

This insight enables a natural extension of the classical greedy paradigm for submodular maximization. Rather than selecting nodes solely based on their marginal contribution to influence, we instead consider a cost-normalized marginal gain that incorporates both the responsiveness of the user and the cost of targeting them. Specifically, we prioritize nodes according to a marginal gain-per-cost criterion scaled by the factor $a_v / w_v$, thereby favoring users who provide the largest increase in expected influence per unit of budget. This formulation preserves the intuitive structure of greedy optimization while explicitly accounting for heterogeneity in both behavioral response and cost, and forms the basis of our algorithmic design.

Building on this intuition, our method proceeds by iteratively allocating budget to nodes in a greedy fashion, guided by a suitably defined marginal gain-per-cost criterion. At a high level, the algorithm incrementally constructs a solution by selecting nodes that offer the greatest increase in expected influence relative to their cost, adjusted for their responsiveness to incentives. 

To formalize and justify this approach, it is important to understand how the continuous influence function $\tilde{\sigma}(\mathbf{x})$ behaves as a function of the allocation vector $\mathbf{x}$. In particular, we require a representation that captures how marginal improvements in incentives translate into expected influence gains under heterogeneous activation functions. The following subsection develops key structural properties of $\tilde{\sigma}(\mathbf{x})$ that underpin our algorithmic design and analysis.

To motivate and analyze our algorithmic approach, we first examine key structural properties of the influence function $\tilde{\sigma}(\mathbf{x})$.

Recall that under the randomized seeding mechanism, the initial active set $S_0$ is a random subset of nodes where each node $v$ is included independently with probability $f_v(x_v)$. Conditioning on $S_0$, the subsequent diffusion follows the IC model, and the total expected influence is given by the discrete influence function $\sigma(S_0)$.

Using this perspective, the continuous influence function can be expressed as
\begin{align}
\tilde{\sigma}(\mathbf{x}) 
= \mathbb{E}_{S_0}\left[\sigma(S_0) \mid \mathbf{x}\right]
= \sum_{S_0 \subseteq V} \sigma(S_0)
\prod_{v \in S_0} f_v(x_v)
\prod_{v \notin S_0} \left(1 - f_v(x_v)\right).
\label{eq:tilde-sigma-expansion}
\end{align}
This representation makes explicit the connection between the continuous formulation and a distribution over discrete seed sets, and will be central to the analysis that follows.

Under the identity activation functions $f_v(x_v) = x_v$, the expression in \eqref{eq:tilde-sigma-expansion} reduces to
\begin{align*}
\tilde{\sigma}(\mathbf{x}) 
= \sum_{S \subseteq V} \sigma(S)
\prod_{v \in S} x_v
\prod_{v \notin S} (1 - x_v),
\end{align*}
which coincides with the standard multilinear extension of $\sigma(\cdot)$.

This establishes a direct connection between the continuous formulation and classical approaches to submodular maximization.

We now consider affine activation functions of the form $f_v(x_v) = a_v x_v + b_v$. Substituting into \eqref{eq:tilde-sigma-expansion}, we observe that $\tilde{\sigma}(\mathbf{x})$ can still be expressed as a weighted sum over subsets $S \subseteq V$, with weights determined by the parameters $\{a_v, b_v\}_{v \in V}$ and the allocation vector $\mathbf{x}$. While the resulting function is no longer the standard multilinear extension of $\sigma(\cdot)$, it retains a related polynomial structure that preserves key features of marginal influence contributions.

This perspective highlights that, although the affine formulation departs from classical multilinear settings, it preserves sufficient structure—particularly in terms of marginal gain behavior—to enable the design of efficient greedy algorithms. These observations provide the foundation for our algorithmic approach in the BGFIM setting.

Leveraging these structural insights, we now develop a discrete-greedy algorithm for solving the BGFIM problem. The algorithm iteratively allocates incentives to users by selecting those that provide the greatest marginal increase in expected influence per unit cost, adjusted for their responsiveness to incentives. This approach extends classical greedy methods for submodular maximization to the present setting, while explicitly accounting for heterogeneity in both activation behavior and targeting costs. We formalize this approach in Algorithm~\ref{alg:discrete-greedy-general}.

\begin{algorithm}[H]
    \caption{Generalized Discrete-Greedy Algorithm for \Cref{problem:bgfim}}
    \label{alg:discrete-greedy-general}
    \textbf{Input:} Users $V$, influence function $\sigma:2^V \rightarrow \mathbb{R}_{\ge 0}$, budget $k$\\
    \textbf{Output:} Allocation vector $y \in \mathbb{R}_{\ge 0}^{|V|}$, where $y_v$ denotes incentive assigned to user $v$
    \begin{algorithmic}[1]
        \State Initialize $y_v \gets 0$ for all $v \in V$ \Comment{No initial incentives to any user}
        \label{alg:step:init-y}
        \State $S_0 \gets \emptyset$ \Comment{Empty initial set of targeted users}
        \label{alg:step:init-s0}
        \While{$\sum_{v \in V} w_v y_v < k$ \textbf{and} $S_0 \neq V$} \Comment{Continue until budget is exhausted or all users are targeted}
        \label{alg:step:while}
            \State $v_{\text{next}} \gets \arg\max_{v \notin S_0} \left( \sigma(S_0 \cup \{v\}) - \sigma(S_0) \right) \cdot \frac{f'_v(0)}{w_v}$
            \Comment{Select a user that maximizes influence gain per unit cost}
            \label{alg:step:choose-max}
            \State $y_{v_{\text{next}}} \gets \min\left( f_{v_{\text{next}}}^{-1}(1), \frac{k - \sum_{v \in V} w_v y_v}{w_{v_{\text{next}}}} \right)$
            \Comment{Allocate maximum feasible incentive within budget}
            \label{alg:step:vnext}
            \State $S_0 \gets S_0 \cup \{v_{\text{next}}\}$
            \Comment{Add user to targeted set}
            \label{alg:step:s0-update}
        \EndWhile
    \end{algorithmic}
\end{algorithm}

\noindent\textbf{Algorithm Description.}
Algorithm~\ref{alg:discrete-greedy-general} formalizes the discrete-greedy strategy for solving \Cref{problem:bgfim}. The method incrementally allocates incentives across users in the social network by selecting, at each step, the user that delivers the highest increase in expected influence per unit of budget. Crucially, this marginal contribution is weighted by the user’s responsiveness to incentives, captured by the local sensitivity of the activation function. This design allows the algorithm to simultaneously account for network effects, heterogeneous adoption behavior, and variation in targeting costs within a unified allocation framework.

\textbf{Initialization (Lines~\ref{alg:step:init-y}--\ref{alg:step:init-s0}).}  
The algorithm begins with no incentives allocated, i.e., $y_v = 0$ for all users $v \in V$, and an empty set $S_0$ representing the users initially targeted in the campaign. At this stage, no adoption has been triggered, and the subsequent allocation decisions are driven entirely by the marginal influence gains generated through the network.

\textbf{Greedy Selection Loop (Lines~\ref{alg:step:while}--\ref{alg:step:s0-update}).}  
At each iteration, the algorithm evaluates the marginal impact of allocating incentives to users who have not yet been targeted. Specifically, it selects the user $v_{\text{next}}$ that maximizes
\[
\left(\sigma(S_0 \cup \{v\}) - \sigma(S_0)\right)\frac{f'_v(0)}{w_v},
\]
which captures the expected increase in overall adoption per unit of budget. This selection rule balances three key considerations: the incremental influence a user can generate through the network, the user's responsiveness to incentives, and the cost required to engage that user. By prioritizing users with the highest influence-per-cost ratio, the algorithm effectively allocates resources toward individuals who can trigger the most impactful diffusion cascades in the social network.

\emph{Set Update (Line~\ref{alg:step:s0-update}).}  
The selected user is added to the targeted set $S_0$, and the allocation process continues. This update reflects the sequential expansion of the campaign, where each newly incentivized user contributes to the growing pool of adopters who can influence others through the social network.

\textbf{Termination.}  
The algorithm terminates when the available campaign budget $k$ is fully utilized or when all users have been considered for targeting. The resulting allocation vector $y$ specifies the incentives assigned across users, effectively defining the final campaign strategy. This allocation balances influence potential, user responsiveness, and cost constraints, yielding a targeted deployment of incentives designed to maximize overall adoption through the social network.

The design of Algorithm~\ref{alg:discrete-greedy-general} closely parallels classical greedy approaches for submodular maximization under budget constraints, while incorporating key adjustments to account for heterogeneous user responsiveness and incentive costs. By prioritizing users based on their cost-adjusted marginal contribution to expected adoption, the algorithm effectively allocates the campaign budget toward individuals with the greatest potential to drive network-wide diffusion. 

Despite the additional complexity introduced by heterogeneous activation functions, we show that this greedy allocation strategy retains strong performance guarantees. In particular, we establish that the proposed algorithm achieves a $(1 - 1/e)$-approximation to the optimal solution of \Cref{problem:bgfim}.

\begin{theorem}[Constant-Factor Approximation Guarantee]
\label{thm:const-approx}
Suppose that the user activation functions are affine, i.e., $f_v(y_v) = a_v y_v + b_v$ with $a_v, b_v \ge 0$ for all users $v \in V$. Then the incentive allocation vector $y$ returned by Algorithm~\ref{alg:discrete-greedy-general} satisfies
\begin{align*}
    F(y) \ge \left(1 - \frac{1}{e} \right) F(y^*),
\end{align*}
where $y^*$ is an optimal solution to \Cref{problem:bgfim}.
\end{theorem}

\begin{proof}{Proof Sketch.}
The proof adapts the classical analysis of greedy algorithms for submodular maximization under budget constraints to the present setting with heterogeneous user behavior. The central idea is to interpret the progression of the campaign as a function of the allocated budget $k$, and to relate the marginal increase in expected influence achieved by the algorithm to the remaining gap between the current allocation and the optimal solution.

At each step, the algorithm selects the user that maximizes the cost-adjusted marginal gain
\[
\left(\sigma(S_0 \cup \{v\}) - \sigma(S_0)\right)\frac{a_v}{w_v},
\]
which can be interpreted as the rate at which expected influence increases per unit of budget. Using the structure of the influence function together with the chain rule, this quantity corresponds to the derivative of the objective $F(y)$ with respect to the budget $k$. Intuitively, the algorithm prioritizes users who deliver the largest marginal increase in adoption for each additional unit of incentive, ensuring that the allocation progresses in the most effective direction at every stage.

To complete the analysis, we compare this rate of improvement with the best possible marginal gain achievable by the optimal allocation $y^*$. Using submodularity of the underlying influence function and the budget constraint, we obtain a bound showing that the rate at which $F(y)$ increases under the algorithm is at least a constant fraction of the remaining optimality gap $F(y^*) - F(y)$. This leads to a differential inequality of the form
\[
\frac{d}{dk} F(y) \ge \frac{1}{k}\left(F(y^*) - F(y)\right),
\]
which captures how the algorithm steadily closes the gap to optimality as the budget is allocated.

Solving this inequality over the interval $[0,k]$ yields
\[
F(y) \ge \left(1 - \frac{1}{e} \right) F(y^*),
\]
establishing the desired $(1 - 1/e)$ approximation guarantee. This matches the classical bound for submodular maximization under knapsack constraints and demonstrates that the proposed algorithm retains strong performance despite the presence of heterogeneous user responses and incentive costs.

\end{proof}

\vspace{1em}
\noindent\textbf{Time Complexity.}
Each iteration of Algorithm~\ref{alg:discrete-greedy-general} selects one user, and thus the procedure runs for at most $|V|$ iterations, where $V$ denotes the set of users in the social network. In each iteration, the algorithm evaluates the marginal influence gain for up to $|V|$ candidate users, each requiring a call to $\sigma(\cdot)$, which can be estimated via Monte Carlo simulation in time $T$. Therefore, the overall time complexity is $O(|V|^2 T)$, matching the complexity of standard greedy methods for influence maximization and enabling scalable implementation for large-scale marketing campaigns.
\vspace{.1in}
\section{Experiments}
\label{sec:experiments}


In this section, we pursue two complementary objectives. First, we evaluate the scalability and effectiveness of the proposed greedy approximation algorithm on large real-world networks spanning multiple application domains. Second, we demonstrate how the nested solutions generated by the algorithm can be leveraged to support budget-selection decisions under varying cost-benefit tradeoffs. Together, these experiments assess both the algorithmic performance of the proposed method and its value as a decision-support tool for influence-based resource allocation.

\subsection{Experimental Setup}

We begin by describing the datasets and preprocessing steps used in our experiments.

\subsubsection{Network Data}
\label{sec:networks}

\begin{table}[h!]
\centering
\begin{tabular}{lrrrr}
\toprule
Network Name & Type & Node Count & Edge Count \\
\midrule
Wikipedia \citep{leskovec2010signed,leskovec2010predicting} & Directed & 7,115   & 201,524  \\
Facebook \citep{leskovec2012learning}  & Undirected & 4,039   & 176,468  \\
Deezer \citep{feather}   & Undirected & 28,281  & 185,504  \\
DBLP  \citep{YangL2012}    & Undirected & 317,080 & 2,099,732  \\
Amazon  \citep{YangL2012}  & Undirected & 334,863 & 1,851,744 \\
\bottomrule
\end{tabular}
\caption{Basic information about each network. The data was retrieved from \citet{snapnets}.}
\label{tab:network-info}
\end{table}

We evaluate our algorithm on five large-scale real-world networks sourced from \citet{snapnets}, each representing a distinct domain of social or economic interaction. Table~\ref{tab:network-info} summarizes their structural properties.
These networks were selected because they represent five distinct influence-allocation settings: governance and voting (Wikipedia), social interaction (Facebook), digital content consumption (Deezer), knowledge diffusion (DBLP), and product adoption (Amazon). Together, they allow us to evaluate whether the proposed algorithm generalizes across fundamentally different behavioral and economic environments.

\begin{enumerate}
    \item \textbf{Wikipedia (Directed):} A voting network where each edge represents a vote cast by one user for another in Wikipedia’s administrator elections. This graph models directional influence and is relevant for studying trust, reputation systems, and decision-making dynamics in online communities.

    \item \textbf{Facebook (Undirected):} A social network of user friendships. As a prototypical example of peer-to-peer influence, this dataset is widely used in studies of viral marketing, information diffusion, and community detection.

    \item \textbf{Deezer (Undirected):} A music streaming platform network where edges represent social connections or shared listening behavior. This network is useful for evaluating influence-based recommendation systems and incentive mechanisms in digital entertainment platforms.

    \item \textbf{DBLP (Undirected):} A co-authorship network in computer science, where edges indicate joint publications. It serves as a model for academic collaboration and the diffusion of ideas, making it relevant for incentive design in research communities.

    \item \textbf{Amazon (Undirected):} A product co-purchase network, where edges connect items frequently bought together. This network is highly applicable to e-commerce, especially for optimizing bundling strategies, targeted promotions, and influence-aware discount allocation.
\end{enumerate}

The diversity of these networks—in terms of size, structure, and domain—allows us to assess the scalability and robustness of our algorithm across different real-world scenarios. Each network presents unique challenges for influence maximization, such as varying degrees of sparsity, clustering, and centrality, which we explore in subsequent sections.

To ensure consistency in influence modeling, we converted all undirected networks into directed graphs by replacing each undirected edge with two directed edges. For edge weights, we adopted the \textit{weighted cascade model} \citep{kempe2003maximizing}, where for each node \( v \in V \), the weight of each incoming edge is set to \( 1/\text{in-degree}(v) \). This normalization ensures that the total incoming influence to any node is bounded by 1, aligning with standard assumptions in influence maximization literature.

We used the \texttt{Pooch} library \citep{Uieda2020} for automated data retrieval and the \texttt{NetworkX} library \citep{networkx} for graph preprocessing and manipulation.

\subsubsection{Experimental Details}
\label{sec:experimental-details}

A central feature of our framework is that different individuals may respond differently to incentives. In practice, customers, users, and agents rarely exhibit identical responsiveness or activation costs. Accordingly, our experimental design incorporates heterogeneous activation and cost parameters.
To simulate realistic incentive allocation scenarios, we randomly assigned each node \( v \in V \) a linear activation function \( f_v(y_v) = a_v y_v + b_v \), where the coefficients \( a_v \), \( b_v \), and cost weights \( w_v \) were drawn uniformly at random from the sets \( A = \{0.5, 1\} \), \( B = \{0, 0.2\} \), and \( W = \{0.5, 1\} \), respectively. These choices reflect heterogeneous agent responsiveness and cost structures, which are common in managerial settings such as marketing campaigns, subsidy allocation, or targeted interventions. For example, some individuals may require higher incentives to reach full activation (low \( a_v \)), while others may be more responsive even with minimal investment (high \( a_v \), low \( b_v \)).

We set the total budget to \( K = 10.0 \), representing a constrained resource environment typical in managerial decision-making. This budget level allows us to evaluate how effectively the algorithm prioritizes high-impact nodes under limited resources.

To estimate the expected influence spread resulting from each discount vector \( y \), we performed 10,000 Monte Carlo simulations per experiment. These simulations model stochastic diffusion processes, capturing the uncertainty inherent in real-world influence propagation.

All experiments were conducted on a workstation equipped with a 3.00 GHz 8-core Intel Xeon E5-1660v3 CPU and 64 GB of RAM. The implementation was written in Python, using the \texttt{CyNetDiff} library \citep{robson2024cynetdiff} for efficient simulation of influence diffusion. To improve runtime performance, we incorporated a priority queue optimization inspired by the Cost-Effective Lazy Forward (CELF) algorithm \citep{celf-07}, which reduces redundant evaluations when selecting the node with the highest marginal gain in step~\ref{alg:step:choose-max} of Algorithm~\ref{alg:discrete-greedy-general}.

\subsection{Greedy Approximation}
\label{sec:greedy-approximation}

We implement Algorithm~\ref{alg:discrete-greedy-general} in Python across all networks described in Section~\ref{sec:networks}. As a baseline for comparison, we use a generalized version of the degree discount heuristic from \citet{chen2009efficient}, where the degree discount score \( dd_v \) is divided by the activation cost \( y_v w_v \). This baseline is chosen for its simplicity and speed, and because degree-based heuristics are widely used in practice due to their intuitive appeal: high-degree nodes often represent popular or central individuals, making them natural targets for influence campaigns.

Figure~\ref{fig:influence_vs_budget} shows the influence spread (estimated via 10,000 Monte Carlo simulations) on the \( y \)-axis as a function of budget usage on the \( x \)-axis. The budget values correspond to the cumulative cost of the selected nodes at each iteration of the algorithm, up to the maximum budget of \( K = 10.0 \). Similarly, Figure~\ref{fig:runtime_vs_budget} presents the runtime (on the $y$-axis) performance of the algorithm as a function of budget (on the $x$-axis).

\begin{figure*}[!htbp]
\centering
\begin{minipage}{.49\textwidth}
  \centering
  \resizebox{\textwidth}{!}{
\begin{tikzpicture}

\definecolor{color0}{rgb}{0.12156862745098,0.466666666666667,0.705882352941177}
\definecolor{color1}{rgb}{1,0.498039215686275,0.0549019607843137}
\definecolor{color2}{rgb}{0.172549019607843,0.627450980392157,0.172549019607843}
\definecolor{color3}{rgb}{0.83921568627451,0.152941176470588,0.156862745098039}
\definecolor{color4}{rgb}{0.580392156862745,0.403921568627451,0.741176470588235}

\begin{axis}[
legend cell align={left},
legend style={fill opacity=0.8, draw opacity=1, text opacity=1, at={(0.30,0.23)}, anchor=north west, draw=white!80!black},
tick align=outside,
tick pos=left,
x grid style={white!69.0196078431373!black},
xlabel={\normalsize Budget},
xmajorgrids,
xmin=-1, xmax=11,
xtick style={color=black},
y grid style={white!69.0196078431373!black},
ylabel={\normalsize Influence},
ymajorgrids,
ymin=-60.401, ymax=2300,
ytick style={color=black}
]

\addplot [thick, only marks, color3, mark=*, mark size=.75, mark options={solid}]
table {%
0.4 138.7526
0.8 274.9891
1.3 419.2109
1.7 534.862
2.1 626.562
3.7 1070.6316
4.2 1165.4736
4.7 1268.8903
5.1 1350.1311
5.5 1421.0731
6.3 1552.7846
6.8 1619.4964
7.2 1675.8072
7.6 1747.3063
8.0 1820.7069
8.4 1872.3978
9.2 1930.2387
9.7 2016.6622
10.0 2043.9987
};
\addlegendentry{Generalized Degree Discount}

\addplot [thick, only marks, color2, mark=*, mark size=1, mark options={solid}]
table {%
0.5 146.6172
0.9 285.3846
1.3 417.9661
1.7 538.4138
3.3 977.9737
3.8 1085.0664
4.8 1278.3486
5.3 1374.5131
5.8 1468.8023
6.2 1558.0135
6.7 1648.1931
7.1 1727.9635
7.6 1801.722
8.1 1877.3709
8.5 1948.9052
8.9 2018.2862
9.3 2093.2247
9.7 2157.8745
10.0 2200.4881
};
\addlegendentry{Generalized Discrete Greedy}
\end{axis}

\end{tikzpicture}}
  \subcaption{Amazon network}
  \label{fig:influence_vs_budget_amazon}
\end{minipage}\hfill
\begin{minipage}{.49\textwidth}
  \centering
  \resizebox{\textwidth}{!}{
\begin{tikzpicture}

\definecolor{color0}{rgb}{0.12156862745098,0.466666666666667,0.705882352941177}
\definecolor{color1}{rgb}{1,0.498039215686275,0.0549019607843137}
\definecolor{color2}{rgb}{0.172549019607843,0.627450980392157,0.172549019607843}
\definecolor{color3}{rgb}{0.83921568627451,0.152941176470588,0.156862745098039}
\definecolor{color4}{rgb}{0.580392156862745,0.403921568627451,0.741176470588235}

\begin{axis}[
legend cell align={left},
legend style={fill opacity=0.8, draw opacity=1, text opacity=1, at={(0.30,0.23)}, anchor=north west, draw=white!80!black},
tick align=outside,
tick pos=left,
x grid style={white!69.0196078431373!black},
xlabel={\normalsize Budget},
xmajorgrids,
xmin=-1, xmax=11,
xtick style={color=black},
y grid style={white!69.0196078431373!black},
ylabel={\normalsize Influence},
ymajorgrids,
ymin=-60.401, ymax=2300,
ytick style={color=black}
]

\addplot [thick, only marks, color3, mark=*, mark size=.75, mark options={solid}]
table {%
0.4 116.5343
0.9 233.198
1.3 376.0986
1.8 490.1933
2.2 582.2295
2.6 683.1858
3.1 790.8291
3.5 882.7664
3.9 927.7811
4.3 1022.7655
4.7 1080.6718
5.1 1116.0098
5.5 1197.314
5.9 1258.6209
6.3 1299.4008
6.8 1396.7307
7.2 1451.421
7.7 1554.82
8.2 1661.3223
8.6 1749.4087
9.1 1843.8982
9.5 1929.7159
10.0 2021.8365
};
\addlegendentry{Generalized Degree Discount}

\addplot [thick, only marks, color2, mark=*, mark size=1, mark options={solid}]
table {%
0.4 143.1265
0.9 262.0574
1.3 376.6702
1.8 489.8562
2.3 602.5764
2.8 713.6026
3.3 825.9009
3.8 934.7915
4.3 1039.1707
4.8 1140.202
5.3 1239.6765
5.7 1336.4168
6.2 1426.4982
6.6 1512.3404
7.0 1598.0832
7.4 1686.1458
7.9 1766.9617
8.3 1853.8528
8.7 1932.9915
9.1 2009.1861
9.5 2098.5828
10.0 2197.0264
};
\addlegendentry{Generalized Discrete Greedy}
\end{axis}

\end{tikzpicture}}
  \subcaption{DBLP network}
  \label{fig:influence_vs_budget_dblp}
\end{minipage}\hfill
\begin{minipage}{.49\textwidth}
  \centering
  \resizebox{\textwidth}{!}{
\begin{tikzpicture}

\definecolor{color0}{rgb}{0.12156862745098,0.466666666666667,0.705882352941177}
\definecolor{color1}{rgb}{1,0.498039215686275,0.0549019607843137}
\definecolor{color2}{rgb}{0.172549019607843,0.627450980392157,0.172549019607843}
\definecolor{color3}{rgb}{0.83921568627451,0.152941176470588,0.156862745098039}
\definecolor{color4}{rgb}{0.580392156862745,0.403921568627451,0.741176470588235}

\begin{axis}[
legend cell align={left},
legend style={fill opacity=0.8, draw opacity=1, text opacity=1, at={(0.30,0.23)}, anchor=north west, draw=white!80!black},
tick align=outside,
tick pos=left,
x grid style={white!69.0196078431373!black},
xlabel={\normalsize Budget},
xmajorgrids,
xmin=-1, xmax=11,
xtick style={color=black},
y grid style={white!69.0196078431373!black},
ylabel={\normalsize Influence},
ymajorgrids,
ymin=-60.401, ymax=1200,
ytick style={color=black}
]

\addplot [thick, only marks, color3, mark=*, mark size=.75, mark options={solid}]
table {%
0.5 141.9324
1.3 263.0296
1.7 319.2145
2.2 383.9258
2.7 457.8356
3.2 525.1309
3.7 581.1807
4.1 621.9981
4.5 672.0293
5.0 722.9947
5.4 760.8912
6.2 839.3634
7.2 913.1695
7.6 948.9841
8.0 979.9733
8.4 1010.7306
8.8 1047.2238
9.3 1089.779
9.7 1123.3113
10.0 1146.7528
};
\addlegendentry{Generalized Degree Discount}

\addplot [thick, only marks, color2, mark=*, mark size=1, mark options={solid}]
table {%
0.5 141.9324
1.0 217.023
1.5 284.7011
2.0 349.6252
2.5 415.8783
2.9 475.5016
3.4 526.5829
4.2 627.2666
4.6 678.8467
5.1 726.0949
5.6 767.6661
6.1 815.5133
6.6 853.7657
7.0 897.2886
7.5 934.5838
7.9 971.761
8.9 1047.2944
9.7 1122.0597
10.0 1151.2557
};
\addlegendentry{Generalized Discrete Greedy}
\end{axis}

\end{tikzpicture}}
  \subcaption{Deezer network}
  \label{fig:influence_vs_budget_deezer}
\end{minipage}\hfill \vspace{.2cm}
\begin{minipage}{.49\textwidth}
  \centering
  \resizebox{\textwidth}{!}{
\begin{tikzpicture}

\definecolor{color0}{rgb}{0.12156862745098,0.466666666666667,0.705882352941177}
\definecolor{color1}{rgb}{1,0.498039215686275,0.0549019607843137}
\definecolor{color2}{rgb}{0.172549019607843,0.627450980392157,0.172549019607843}
\definecolor{color3}{rgb}{0.83921568627451,0.152941176470588,0.156862745098039}
\definecolor{color4}{rgb}{0.580392156862745,0.403921568627451,0.741176470588235}

\begin{axis}[
legend cell align={left},
legend style={fill opacity=0.8, draw opacity=1, text opacity=1, at={(0.30,0.23)}, anchor=north west, draw=white!80!black},
tick align=outside,
tick pos=left,
x grid style={white!69.0196078431373!black},
xlabel={\normalsize Budget},
xmajorgrids,
xmin=-1, xmax=11,
xtick style={color=black},
y grid style={white!69.0196078431373!black},
ylabel={\normalsize Influence},
ymajorgrids,
ymin=-20.401, ymax=1000,
ytick style={color=black}
]

\addplot [thick, only marks, color3, mark=*, mark size=.75, mark options={solid}]
table {%
1.0 190.9579
1.8 343.5229
2.8 449.3325
3.3 559.3073
3.7 585.5215
4.7 730.6466
5.1 740.9732
5.5 755.9485
5.9 763.8243
6.4 776.3716
6.8 784.7199
7.2 792.4851
7.6 804.5672
8.0 809.0769
8.4 814.9213
8.8 822.7007
9.2 827.5121
9.6 870.902
10.0 878.8188
};
\addlegendentry{Generalized Degree Discount}

\addplot [thick, only marks, color2, mark=*, mark size=1, mark options={solid}]
table {%
0.5 111.9272
1.5 300.8041
2.3 452.0764
3.3 597.8661
4.3 704.33
4.7 748.2278
5.1 773.5452
5.9 817.2961
6.7 846.8628
7.2 861.0002
8.2 888.3838
8.7 901.6435
9.1 913.1884
9.5 924.655
9.9 936.3339
10.0 940.3371
};
\addlegendentry{Generalized Discrete Greedy}
\end{axis}

\end{tikzpicture}}
  \subcaption{Facebook network}
  \label{fig:influence_vs_budget_facebook}
\end{minipage}\hfill
\begin{minipage}{.49\textwidth}
  \centering
  \resizebox{\textwidth}{!}{
\begin{tikzpicture}

\definecolor{color0}{rgb}{0.12156862745098,0.466666666666667,0.705882352941177}
\definecolor{color1}{rgb}{1,0.498039215686275,0.0549019607843137}
\definecolor{color2}{rgb}{0.172549019607843,0.627450980392157,0.172549019607843}
\definecolor{color3}{rgb}{0.83921568627451,0.152941176470588,0.156862745098039}
\definecolor{color4}{rgb}{0.580392156862745,0.403921568627451,0.741176470588235}

\begin{axis}[
legend cell align={left},
legend style={fill opacity=0.8, draw opacity=1, text opacity=1, at={(0.30,0.23)}, anchor=north west, draw=white!80!black},
tick align=outside,
tick pos=left,
x grid style={white!69.0196078431373!black},
xlabel={\normalsize Budget},
xmajorgrids,
xmin=-1, xmax=11,
xtick style={color=black},
y grid style={white!69.0196078431373!black},
ylabel={\normalsize Influence},
ymajorgrids,
ymin=-20.401, ymax=1800,
ytick style={color=black}
]

\addplot [thick, only marks, color3, mark=*, mark size=.75, mark options={solid}]
table {%
0.5 306.4505
1.0 525.1954
1.8 771.4566
2.2 894.091
2.6 963.9116
3.0 1022.0489
3.4 1079.4412
3.8 1130.6578
4.3 1195.1802
4.7 1240.9265
5.5 1318.7732
5.9 1347.471
6.3 1394.9391
7.1 1449.4469
7.6 1487.4051
8.1 1520.2967
8.6 1554.4086
9.1 1593.5518
9.5 1614.0678
10.0 1641.5855
};
\addlegendentry{Generalized Degree Discount}

\addplot [thick, only marks, color2, mark=*, mark size=1, mark options={solid}]
table {%
0.5 306.4505
1.0 525.1954
1.4 678.5509
2.2 891.7527
2.7 974.2872
3.1 1038.5718
3.6 1092.7576
4.1 1151.7668
4.5 1203.4941
4.9 1248.9055
5.3 1294.8498
5.7 1336.5875
6.1 1374.3027
6.6 1414.5996
7.1 1450.1435
7.5 1486.3352
8.3 1544.5418
8.8 1576.1541
9.6 1626.7416
10.0 1664.0846
};
\addlegendentry{Generalized Discrete Greedy}
\end{axis}

\end{tikzpicture}}
  \subcaption{Wikipedia network}
  \label{fig:influence_vs_budget_wikipedia}
\end{minipage} \vspace{.1in}
\caption{Influence vs. Budget ($K$) for different networks. 
}
\label{fig:influence_vs_budget}
\end{figure*}
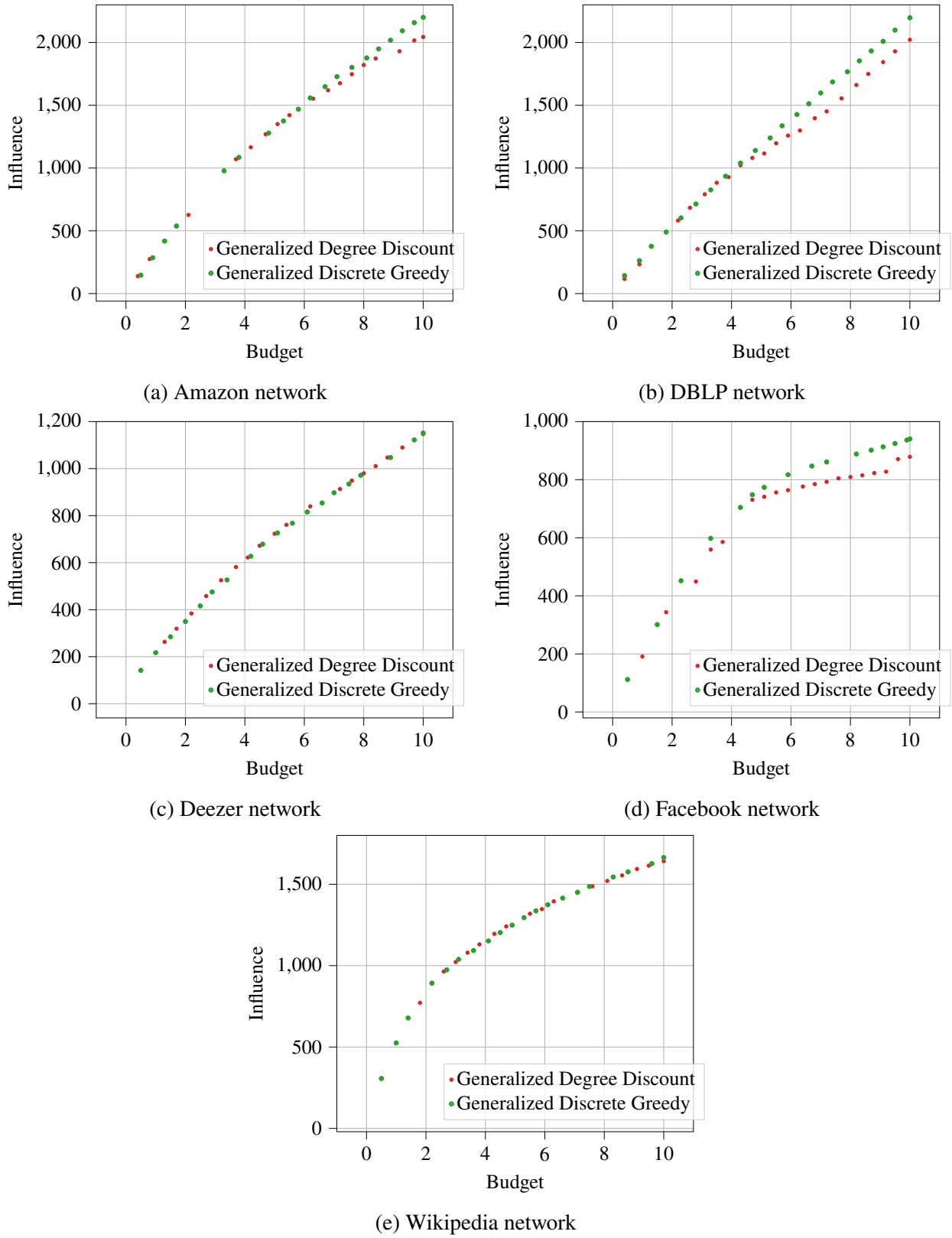

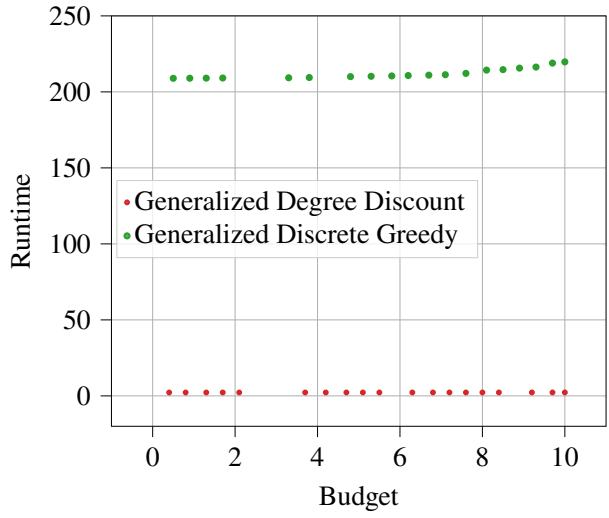
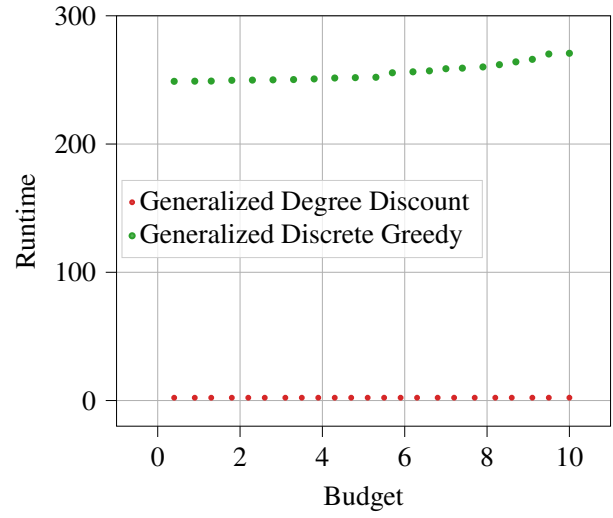
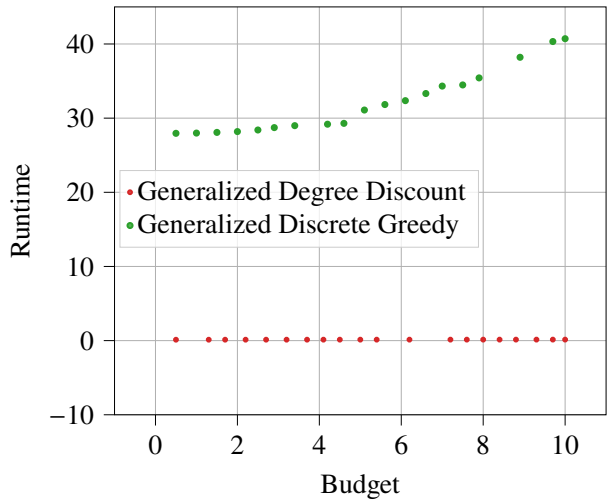
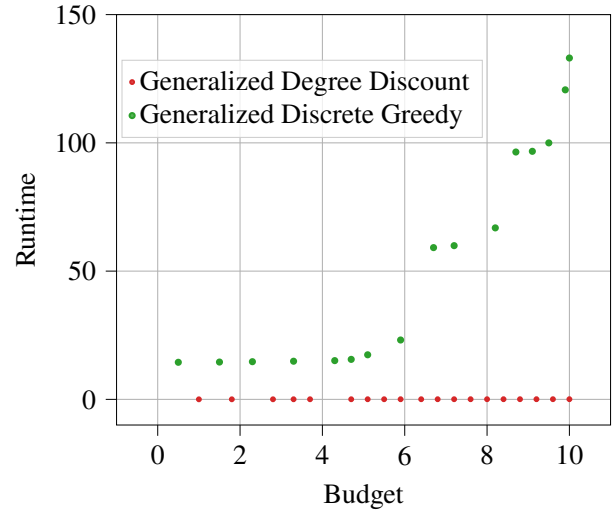
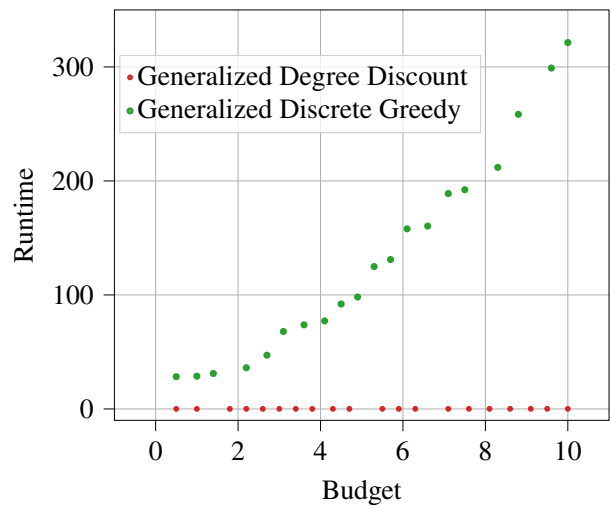
\begin{figure*}[!htbp]
\centering
\begin{minipage}{.49\textwidth}
  \centering
  \resizebox{\textwidth}{!}{
\begin{tikzpicture}

\definecolor{color0}{rgb}{0.12156862745098,0.466666666666667,0.705882352941177}
\definecolor{color1}{rgb}{1,0.498039215686275,0.0549019607843137}
\definecolor{color2}{rgb}{0.172549019607843,0.627450980392157,0.172549019607843}
\definecolor{color3}{rgb}{0.83921568627451,0.152941176470588,0.156862745098039}
\definecolor{color4}{rgb}{0.580392156862745,0.403921568627451,0.741176470588235}

\begin{axis}[
legend cell align={left},
legend style={
  fill opacity=0.8,
  draw opacity=1,
  text opacity=1,
  at={(0.01,0.60)},
  anchor=north west,
  draw=white!80!black
},
log basis y={10},
tick align=outside,
tick pos=left,
x grid style={white!69.0196078431373!black},
xlabel={\normalsize Budget},
xmajorgrids,
xmin=-1, xmax=11,
xtick style={color=black},
y grid style={white!69.0196078431373!black},
ylabel={\normalsize Runtime},
ymajorgrids,
ymin=-20, ymax=250,
ytick style={color=black}
]

\addplot [thick, only marks, color3, mark=*, mark size=.75, mark options={solid}]
table {%
0.4 2.2325717310013715
0.8 2.234879935000208
1.3 2.2370541669952217
1.7 2.239544969997951
2.1 2.2412173560005613
3.7 2.242619406999438
4.2 2.2489295709965518
4.7 2.250537704996532
5.1 2.252096601005178
5.5 2.25345132300572
6.3 2.2547627379972255
6.8 2.257213036995381
7.2 2.2586830189975444
7.6 2.2599542300013127
8.0 2.2611440680047963
8.4 2.2623435290006455
9.2 2.263534940997488
9.7 2.265884770997218
10.0 2.2673209000058705
};
\addlegendentry{Generalized Degree Discount}

\addplot [thick, only marks, color2, mark=*, mark size=1, mark options={solid}]
table {%
0.5 208.93439140700502
0.9 208.97764774999814
1.3 209.0423374969978
1.7 209.12926917499863
3.3 209.27603863300465
3.8 209.44248785100353
4.8 210.0312739719957
5.3 210.25493404400186
5.8 210.48909080600424
6.2 210.74563532900356
6.7 211.01144936100172
7.1 211.29310377499496
7.6 212.16869065299397
8.1 214.32558830799826
8.5 214.6427277259936
8.9 215.63208882999606
9.3 216.329580687001
9.7 218.93720283100265
10.0 219.73839917700388
};
\addlegendentry{Generalized Discrete Greedy}
\end{axis}

\end{tikzpicture}}
  \subcaption{Amazon network}
  \label{fig:runtime_vs_budget_Amazon}
\end{minipage}\hfill
\begin{minipage}{.49\textwidth}
  \centering
  \resizebox{\textwidth}{!}{
\begin{tikzpicture}

\definecolor{color0}{rgb}{0.12156862745098,0.466666666666667,0.705882352941177}
\definecolor{color1}{rgb}{1,0.498039215686275,0.0549019607843137}
\definecolor{color2}{rgb}{0.172549019607843,0.627450980392157,0.172549019607843}
\definecolor{color3}{rgb}{0.83921568627451,0.152941176470588,0.156862745098039}
\definecolor{color4}{rgb}{0.580392156862745,0.403921568627451,0.741176470588235}

\begin{axis}[
legend cell align={left},
legend style={
  fill opacity=0.8,
  draw opacity=1,
  text opacity=1,
  at={(0.01,0.60)},
  anchor=north west,
  draw=white!80!black
},
log basis y={10},
tick align=outside,
tick pos=left,
x grid style={white!69.0196078431373!black},
xlabel={\normalsize Budget},
xmajorgrids,
xmin=-1, xmax=11,
xtick style={color=black},
y grid style={white!69.0196078431373!black},
ylabel={\normalsize Runtime},
ymajorgrids,
ymin=-20, ymax=300,
ytick style={color=black}
]

\addplot [thick, only marks, color3, mark=*, mark size=.75, mark options={solid}]
table {%
0.4 2.192481383011909
0.9 2.194915377011057
1.3 2.1979289070004597
1.8 2.200297343006241
2.2 2.202471239012084
2.6 2.204479496998829
3.1 2.207166773005156
3.5 2.2093365760083543
3.9 2.2108871689997613
4.3 2.212610804999713
4.7 2.214234099999885
5.1 2.215840447010123
5.5 2.2171800720097963
5.9 2.218704972998239
6.3 2.22019533799903
6.8 2.2216066510009114
7.2 2.2235359820042504
7.7 2.2249512569978833
8.2 2.2268513050075853
8.6 2.2286976920004236
9.1 2.2302351890102727
9.5 2.2320026850065915
10.0 2.2334806120052235
};
\addlegendentry{Generalized Degree Discount}

\addplot [thick, only marks, color2, mark=*, mark size=1, mark options={solid}]
table {%
0.4 248.89323503201012
0.9 249.01002006600902
1.3 249.09462065900152
1.8 249.6874783950043
2.3 249.85613092900894
2.8 250.04419579100795
3.3 250.27348011100548
3.8 250.7707109710027
4.3 251.48258150400943
4.8 251.76762331000646
5.3 252.05793206000817
5.7 255.5696720870037
6.2 256.3145575640083
6.6 257.06822782701056
7.0 258.72635837300913
7.4 259.18095814299886
7.9 260.1114929000032
8.3 261.8796031390084
8.7 264.07014085300034
9.1 266.042541599003
9.5 270.19685859201127
10.0 270.7678656250064
};
\addlegendentry{Generalized Discrete Greedy}
\end{axis}

\end{tikzpicture}}
  \subcaption{DBLP network}
  \label{fig:runtime_vs_budget_DBLP}
\end{minipage}\hfill
\begin{minipage}{.49\textwidth}
  \centering
  \resizebox{\textwidth}{!}{
\begin{tikzpicture}

\definecolor{color0}{rgb}{0.12156862745098,0.466666666666667,0.705882352941177}
\definecolor{color1}{rgb}{1,0.498039215686275,0.0549019607843137}
\definecolor{color2}{rgb}{0.172549019607843,0.627450980392157,0.172549019607843}
\definecolor{color3}{rgb}{0.83921568627451,0.152941176470588,0.156862745098039}
\definecolor{color4}{rgb}{0.580392156862745,0.403921568627451,0.741176470588235}

\begin{axis}[
legend cell align={left},
legend style={
  fill opacity=0.8,
  draw opacity=1,
  text opacity=1,
  at={(0.01,0.60)},
  anchor=north west,
  draw=white!80!black
},
log basis y={10},
tick align=outside,
tick pos=left,
x grid style={white!69.0196078431373!black},
xlabel={\normalsize Budget},
xmajorgrids,
xmin=-1, xmax=11,
xtick style={color=black},
y grid style={white!69.0196078431373!black},
ylabel={\normalsize Runtime},
ymajorgrids,
ymin=-10, ymax=45,
ytick style={color=black}
]

\addplot [thick, only marks, color3, mark=*, mark size=.75, mark options={solid}]
table {%
0.5 0.12210151299950667
1.3 0.1237163970072288
1.7 0.12509476600098424
2.2 0.12577266999869607
2.7 0.126644457996008
3.2 0.12749502299993765
3.7 0.12828876299317926
4.1 0.12902203599514905
4.5 0.12968000299588311
5.0 0.13024321499688085
5.4 0.13098262599669397
6.2 0.13151271200331394
7.2 0.13255942400428466
7.6 0.13377671099442523
8.0 0.13432887400267646
8.4 0.13480124999477994
8.8 0.13532253600715194
9.3 0.13583428900165018
9.7 0.1364285559975542
10.0 0.13690578800742514
};
\addlegendentry{Generalized Degree Discount}

\addplot [thick, only marks, color2, mark=*, mark size=1, mark options={solid}]
table {%
0.5 27.952514258999145
1.0 27.9861977929977
1.5 28.07651995299966
2.0 28.18910582200624
2.5 28.401857801000006
2.9 28.7236892140063
3.4 28.987164133999613
4.2 29.18244457100809
4.6 29.296273605010356
5.1 31.09345271899656
5.6 31.83488687699719
6.1 32.35885288999998
6.6 33.31974787900981
7.0 34.32321976601088
7.5 34.47609730200202
7.9 35.42564212200523
8.9 38.206328051004675
9.7 40.33126465900568
10.0 40.71220151100715
};
\addlegendentry{Generalized Discrete Greedy}
\end{axis}

\end{tikzpicture}}
  \subcaption{Deezer network}
  \label{fig:runtime_vs_budget_Deezer}
\end{minipage}\hfill \vspace{.2cm}
\begin{minipage}{.49\textwidth}
  \centering
  \resizebox{\textwidth}{!}{
\begin{tikzpicture}

\definecolor{color0}{rgb}{0.12156862745098,0.466666666666667,0.705882352941177}
\definecolor{color1}{rgb}{1,0.498039215686275,0.0549019607843137}
\definecolor{color2}{rgb}{0.172549019607843,0.627450980392157,0.172549019607843}
\definecolor{color3}{rgb}{0.83921568627451,0.152941176470588,0.156862745098039}
\definecolor{color4}{rgb}{0.580392156862745,0.403921568627451,0.741176470588235}

\begin{axis}[
legend cell align={left},
legend style={
  fill opacity=0.8,
  draw opacity=1,
  text opacity=1,
  at={(0.01,0.89)},
  anchor=north west,
  draw=white!80!black
},
log basis y={10},
tick align=outside,
tick pos=left,
x grid style={white!69.0196078431373!black},
xlabel={\normalsize Budget},
xmajorgrids,
xmin=-1, xmax=11,
xtick style={color=black},
y grid style={white!69.0196078431373!black},
ylabel={\normalsize Runtime},
ymajorgrids,
ymin=-10, ymax=150,
ytick style={color=black}
]

\addplot [thick, only marks, color3, mark=*, mark size=.75, mark options={solid}]
table {%
1.0 0.017676257994025946
1.8 0.02594061200215947
2.8 0.03208066300430801
3.3 0.038051176990848035
3.7 0.040858747001038864
4.7 0.04266807599924505
5.1 0.04683493399352301
5.5 0.04845420300262049
5.9 0.05017695500282571
6.4 0.051773707004031166
6.8 0.05382792599266395
7.2 0.055450735992053524
7.6 0.057052018004469573
8.0 0.05859377399610821
8.4 0.06017344999418128
8.8 0.06168924599478487
9.2 0.0631384049920598
9.6 0.06465276300150435
10.0 0.0660179539991077
};
\addlegendentry{Generalized Degree Discount}

\addplot [thick, only marks, color2, mark=*, mark size=1, mark options={solid}]
table {%
0.5 14.454121381000732
1.5 14.546448990004137
2.3 14.690919989006943
3.3 14.870818046998465
4.3 15.10462761000963
4.7 15.592257338998024
5.1 17.37628255900927
5.9 23.14389065100113
6.7 59.16128571200534
7.2 59.92711412500648
8.2 66.84973878100573
8.7 96.43239662901033
9.1 96.70534216500528
9.5 99.98782677500276
9.9 120.6707146550034
10.0 133.08097400399856
};
\addlegendentry{Generalized Discrete Greedy}
\end{axis}

\end{tikzpicture}}
  \subcaption{Facebook network}
  \label{fig:runtime_vs_budget_Facebook}
\end{minipage}\hfill
\begin{minipage}{.49\textwidth}
  \centering
  \resizebox{\textwidth}{!}{
\begin{tikzpicture}

\definecolor{color0}{rgb}{0.12156862745098,0.466666666666667,0.705882352941177}
\definecolor{color1}{rgb}{1,0.498039215686275,0.0549019607843137}
\definecolor{color2}{rgb}{0.172549019607843,0.627450980392157,0.172549019607843}
\definecolor{color3}{rgb}{0.83921568627451,0.152941176470588,0.156862745098039}
\definecolor{color4}{rgb}{0.580392156862745,0.403921568627451,0.741176470588235}

\begin{axis}[
legend cell align={left},
legend style={
  fill opacity=0.8,
  draw opacity=1,
  text opacity=1,
  at={(0.01,0.89)},
  anchor=north west,
  draw=white!80!black
},
log basis y={10},
tick align=outside,
tick pos=left,
x grid style={white!69.0196078431373!black},
xlabel={\normalsize Budget},
xmajorgrids,
xmin=-1, xmax=11,
xtick style={color=black},
y grid style={white!69.0196078431373!black},
ylabel={\normalsize Runtime},
ymajorgrids,
ymin=-10, ymax=350,
ytick style={color=black}
]

\addplot [thick, only marks, color3, mark=*, mark size=.75, mark options={solid}]
table {%
0.5 0.030454884996288456
1.0 0.036734220004291274
1.8 0.042493971006479114
2.2 0.051016395998885855
2.6 0.05475395200483035
3.0 0.05817162900348194
3.4 0.06130794399359729
3.8 0.06451164699683432
4.3 0.06761848299356643
4.7 0.07142068599932827
5.5 0.0742737069958821
5.9 0.07957963900116738
6.3 0.08221492200391367
7.1 0.08458437799708918
7.6 0.0893422700028168
8.1 0.09230758600460831
8.6 0.09517254499951378
9.1 0.09789409300719853
9.5 0.10045884099963587
10.0 0.10242025200568605
};
\addlegendentry{Generalized Degree Discount}

\addplot [thick, only marks, color2, mark=*, mark size=1, mark options={solid}]
table {%
0.5 28.302392787998542
1.0 28.724055031998432
1.4 31.09117438699468
2.2 36.14526251000643
2.7 47.14117938300478
3.1 67.91225550600211
3.6 73.7577517220052
4.1 77.22335173700412
4.5 92.08463623300486
4.9 98.17050827700587
5.3 124.88954313099384
5.7 131.0519708770007
6.1 157.97586084999784
6.6 160.38594619900687
7.1 188.93136124000011
7.5 192.27463532499678
8.3 211.89284415599832
8.8 258.41033494099975
9.6 299.02162606900674
10.0 321.3704178140033
};
\addlegendentry{Generalized Discrete Greedy}
\end{axis}

\end{tikzpicture}}
  \subcaption{Wikipedia network}
  \label{fig:runtime_vs_budget_Wikipedia}
\end{minipage} \vspace{.1in}

\caption{Runtime (in seconds) vs. Budget ($K$) for different networks. 
}
\label{fig:runtime_vs_budget}
\end{figure*}

These results generate several managerial insights:
\begin{itemize}
    \item \textbf{Strategic Resource Allocation:} The greedy algorithm provides a principled way to allocate limited incentives across a population to maximize influence. This is directly applicable to domains such as marketing (e.g., targeting influential customers), public health (e.g., vaccine outreach), and platform growth (e.g., onboarding key users in a network).

    \item \textbf{Scalability and Efficiency:} The use of CELF-style optimizations \citep{celf-07} ensures that the algorithm remains computationally tractable even on large-scale networks. This is crucial for managers who need to make timely decisions under resource constraints.

    \item \textbf{Performance vs. Heuristics:} Compared to the generalized degree discount heuristic, our algorithm consistently achieves higher influence, especially at larger budget levels. This demonstrates the value of using theoretically grounded methods over simple heuristics, particularly when the stakes of resource allocation are high.

    \item \textbf{Transparency and Interpretability:} The greedy algorithm’s step-by-step selection process offers transparency in decision-making, allowing managers to justify why certain individuals or nodes were prioritized for incentives.

\end{itemize}

Overall, our framework bridges the gap between algorithmic rigor and managerial applicability, offering a scalable and interpretable solution for influence maximization under budget constraints.

\subsection{Profit-Optimized Budget Selection}

The experiments thus far have focused on influence maximization under a fixed budget. In practice, however, managers often face the inverse problem: determining how much budget should be allocated in the first place. This requires balancing the additional influence generated by larger budgets against the corresponding increase in incentive costs. Because our greedy algorithm produces a nested sequence of solutions, the resulting influence estimates can be reused across many candidate budget levels, enabling efficient profit-based decision-making.

In real-world marketing campaigns, the goal is not merely to maximize influence but to do so in a way that maximizes profit. Each unit of budget spent on incentives incurs a cost, while each successful activation generates revenue. The challenge for marketers is to determine the optimal budget allocation that balances these two forces—maximizing reach without overspending.

Let $p \ge 0$ denote the price per unit of the product being promoted, and $c \ge 0$ the cost of sponsoring one unit. Then, the profit $\pi(t)$ from sponsoring $t$ units is given by:
\begin{align}
    \text{total profit} &= \text{total revenue} - \text{total cost of sponsorship}, \nonumber\\
    &= \text{price per unit} \times \text{influence} - \text{cost per unit} \times \text{budget}, \nonumber\\
    \pi(t) &= p\tilde{\sigma}(\mathbf{x}(t)) - ct
\end{align}

Here, $\tilde{\sigma}(\mathbf{x}(t))$ represents the expected influence achieved by incentivizing $t$ users, selected using our greedy algorithm.



To identify the budget \(t\) that maximizes profit \(\pi(t)\), we use \textit{ternary search}. Ternary search is a standard optimization procedure for unimodal functions, i.e., functions that increase up to a single maximum and then decrease. At each iteration, the method evaluates two interior points within the current budget interval and eliminates the region that cannot contain the optimum. By repeatedly shrinking the search interval, ternary search efficiently identifies the profit-maximizing budget.


Importantly, our greedy algorithm records the order in which nodes are selected, allowing influence to be evaluated incrementally and avoiding recomputation. Once the influence dictionary has been constructed, managers can rapidly evaluate alternative pricing and cost scenarios without rerunning the underlying influence-maximization algorithm. This effectively separates the computationally intensive influence-maximization step from the subsequent budget-selection problem, making profit optimization computationally inexpensive.

\begin{table}[h]
\centering
\begin{tabular}{lrrrrr}
\toprule
Graph & Budget & Payoff & Search Time & Initial Computation Time \\
\midrule
Wikipedia & 2.82  & 275.54 & 63.21 & 212.30 \\
Facebook  & 1.58  &  45.57 & 25.94 &  61.99 \\
Deezer    & 0.59  &  29.90 &  9.12 &  39.34 \\
DBLP      & 7.65  & 250.41 & 75.13 & 265.84 \\
Amazon    & 7.14  & 295.35 & 60.50 & 219.26 \\
\bottomrule
\end{tabular}
\caption{Results of optimal budget selection across networks. Experiments used a maximum budget of $K = 8.00$, price per unit $p = 0.5$, cost per unit $c = 80.0$, and $\epsilon = 0.1$.}
\label{tab:graph-budget-payoff-120}
\end{table}


\Cref{tab:graph-budget-payoff-120} presents the results of our profit-optimized budget selection experiments across five real-world networks: Wikipedia, Facebook, Deezer, DBLP, and Amazon. For each graph, we report four key metrics:
\begin{enumerate}
    \item Optimal Budget: The budget value that maximizes profit under the given pricing and cost parameters.
    \item Payoff: The resulting profit achieved at the optimal budget.
    \item Search Time: The time taken to perform a ternary search over the budget space to identify the profit-maximizing point.
    \item Initial Computation Time: The time required to compute the influence dictionary up to the maximum budget $K = 8.00$.
\end{enumerate}

These experiments were conducted using a fixed price per unit of $p = 0.5$, a cost per unit of $c = 80.0$, and a fractional evaluation precision of $\epsilon = 0.1$.
The results reveal several important trends:
\begin{itemize}
    \item \textbf{Efficiency of Optimization:} Across all networks, the ternary search procedure is significantly faster than the initial influence computation. For example, in the Wikipedia network, the search time (63.21 seconds) is less than one-third of the initial computation time (212.30 seconds). This demonstrates that once the influence landscape is computed, profit optimization becomes highly efficient and scalable.

    \item \textbf{Network-Specific Budgets:} The optimal budget varies significantly across networks, reflecting differences in network structure and user responsiveness. For instance, Deezer achieves its maximum profit with a budget of just 0.59, while DBLP requires a much larger budget of 7.65. This highlights the importance of tailoring budget strategies to the underlying network characteristics.

    \item \textbf{Profit Sensitivity:} The payoff values also vary widely, with Amazon yielding the highest profit (295.35) and Deezer the lowest (29.90). These differences underscore the impact of network topology and influence dynamics on campaign profitability.

    \item \textbf{Managerial Implication:} From a business perspective, these results validate the practicality of our framework. Marketing managers can precompute influence dictionaries for their target networks and then rapidly evaluate optimal budget levels under different pricing and cost scenarios. This enables data-driven decision-making and supports agile campaign planning in dynamic environments.
\end{itemize}

Overall, this case study demonstrates that our framework not only provides theoretical guarantees and algorithmic efficiency but also delivers actionable insights for real-world marketing applications, offering a clear and actionable pathway from data to profit.


\vspace{.1in}
\section{Conclusion}
\label{sec:conclusion}

We conclude the paper by summarizing our contributions, discussing the framework’s industry readiness, outlining its limitations, and identifying promising directions for future research.

\subsection{Summary}
In this work, we addressed the problem of influence maximization using partial incentives under a generalized marketing model. We formally defined the Budgeted Generalized Fractional Influence Maximization (BGFIM) problem, which extends prior work by incorporating two key real-world factors: (i) a flexible affine relationship between the offered discount and a user's probability of activation, and (ii) a generalized budget constraint that accounts for heterogeneous costs in incentivizing different users.

To solve this more realistic and complex problem, we proposed an efficient greedy algorithm with a provable approximation guarantee of $\left(1 - 1/e\right)$ relative to the optimal solution. Through extensive experiments on five large-scale real-world networks, we demonstrated the algorithm's scalability and superior performance compared to common heuristics. Additionally, we showcased its utility in a profit-maximization setting, highlighting its relevance for strategic marketing decisions.

\subsection{Industry Readiness}
Our framework is well-positioned for real-world deployment in marketing and advertising platforms. The ability to model heterogeneous user costs and partial incentives aligns with how modern platforms operate—offering personalized discounts, targeting high-value users, and optimizing campaign budgets. The profit-optimized budget selection method enables marketers to make data-driven decisions that balance cost and reach, maximizing return on investment (ROI). Moreover, the reuse of influence computations across different pricing scenarios makes the approach scalable and cost-effective for dynamic business environments.

\subsection{Limitations}
While our model generalizes prior work, it still assumes a linear relationship between discount and activation probability. In practice, user behavior may follow more complex patterns—such as diminishing returns or sigmoid responses—where small discounts have little effect until a threshold is reached, after which influence grows rapidly and then plateaus. Our current framework does not capture such non-linear dynamics. Additionally, the cost of incentivizing users is treated as static, whereas in real-world settings, costs may evolve over time due to factors like user popularity 

\subsection{Future Steps}
Building on this foundation, several promising directions emerge. First, we aim to extend the model to support non-linear activation functions, such as convex or sigmoid relationships, which better reflect real-world consumer behavior. Second, we plan to explore dynamic cost models where the cost of incentivizing a user may change over time, enabling more adaptive campaign strategies. Finally, we are interested in developing online algorithms for influence maximization in uncertain environments, where activation probabilities are not known in advance and must be learned over time through exploration and feedback \citep{agarwal2022stochastic,nie2022explore, xu2026lofa}. These extensions will further enhance the applicability of our framework to real-time, data-driven marketing systems.

\vspace{.1in}
\section*{Acknowledgments} \label{sec:acknowledgments}
A preliminary version of this paper appeared at the 33rd ACM International Conference on Information and Knowledge Management (CIKM), 2024 \citep{BhimarajuRVU2024}. The present manuscript substantially extends and revises that earlier version.

\vspace{.1in}
\bibliographystyle{informs2014}
\bibliography{refs.bib}
\end{document}